\documentclass[sigplan,screen]{acmart}

\usepackage{mathtools}

\usepackage{booktabs}
\usepackage{microtype}   
\usepackage{graphicx}
\usepackage{xspace}
\usepackage{xcolor}
\usepackage{tikz}
\usetikzlibrary{arrows.meta,positioning,shapes.geometric,calc}
\usepackage{pgfplots}
\pgfplotsset{compat=1.18}
\usepgfplotslibrary{groupplots}
\pgfplotsset{
  paperaxis/.style={
    axis lines*=left,
    axis line style={draw=grid},
    tick style={draw=grid},
    tick label style={font=\scriptsize, text=ink},
    label style={font=\footnotesize, text=ink},
    ymajorgrids,
    grid style={grid!55},
  }
}

\definecolor{blue}{HTML}{2A78D6}
\definecolor{red}{HTML}{E34948}
\definecolor{green}{HTML}{1BAF7A}
\definecolor{orange}{HTML}{E58B3A}
\definecolor{purple}{HTML}{4A3AA7}
\definecolor{yellow}{HTML}{E5B84B}
\definecolor{ink}{HTML}{52514E}
\definecolor{grid}{HTML}{D8D7D2}

\newcommand{\aufbau}{\textsc{Aufbau}\xspace}
\newcommand{\pseven}{P7\xspace}
\newcommand{\ptool}{\pseven}
\newcommand{\tryfeed}{\texttt{try\_feed}\xspace}
\newcommand{\Lost}{\texttt{Lost}\xspace}
\newcommand{\Live}{\texttt{Live}\xspace}
\newcommand{\Satisfied}{\texttt{Satisfied}\xspace}
\newcommand{\evalf}{\mathit{eval}}
\newcommand{\hole}[1]{{\ensuremath{\mathfrak{#1}}}}

\setcopyright{cc}
\setcctype{by}
\acmDOI{10.1145/3843750.3843841}
\acmYear{2026}
\copyrightyear{2026}
\acmISBN{979-8-4007-2986-7/2026/10}
\acmConference[LMPL '26]{Proceedings of the 2nd ACM SIGPLAN International Workshop on Language Models and Programming Languages}{October 4--9, 2026}{Oakland, CA, USA}
\acmBooktitle{Proceedings of the 2nd ACM SIGPLAN International Workshop on Language Models and Programming Languages (LMPL '26), October 4--9, 2026, Oakland, CA, USA}
\acmSubmissionID{splashws26lmplmain-p30-p}
\received{2026-07-10}
\received[accepted]{2026-08-14}

\title{Semantic Prefix Oracles for LLM Decoding: Contracts and Differential Validation}

\author{Paul Kronlund-Drouault}
\correspondingauthor
\orcid{0009-0001-9326-4878}

\affiliation{%
  \institution{ENS de Lyon}
  \city{Lyon}
  \country{France}
}

\affiliation{%
  \institution{Unsuspicious Industries}
  \city{Paris}
  \country{France}
}

\affiliation{%
  \institution{Université de Lille}
  \city{Lille}
  \country{France}
}

\email{pkd@unsuspicious.org}

\begin{document}

\begin{abstract}
Constrained decoding can enforce regular or context-free output formats, but many program-generation failures are semantic: scope, typing, and declaration effects depend on context. We present semantic grammar specifications, a declarative formalism that attaches such constraints to a context-free surface and executes them during Earley descent. Our implementation enforces \emph{safe pruning}: it rejects only prefixes whose semantic contradictions cannot be repaired by any continuation. A separate, grammar-dependent, \emph{dead-end freedom} property guarantees the existence of a realizable witness for each remaining branch. We give simple sufficient conditions based on surface productivity, type coverage, and left-to-right constraint flow. Our finite-lambda, core ML, and C-like fragments satisfy them, while the STLC instance used in our experiments does not: plain STLC can violate type coverage, and we show how restricting its type universe recovers it. A tokenizer-lifting lemma carries character-level witnesses to token sequences under an explicit vocabulary-coverage hypothesis.

We validate the implementation differentially against production compilers (\texttt{ocamlc}, \texttt{cc}). Across every prefix of 65 compiler-valid programs we observe zero false prunes. The semantic oracle localizes 25/30 invalid programs mid-stream, against 0/30 for a syntax-only oracle, and agrees on 42/42 recursion probes. A twelve-model generation study, including a matched semantic-versus-syntactic ablation for nine models, finds nonnegative observed semantic-minus-syntactic point estimates for every model-language pair, with maxima of $+15.2$ points on STLC task correctness and $+14.3$ points on ML validity.
\end{abstract}



\maketitle

\section{Introduction}\label{sec:intro}

Grammar-constrained decoding has by now entered inference infrastructure, with libraries such as Outlines~\cite{willard2023efficient}, DOMINO~\cite{beurerkellner2024guiding}, and Formatron~\cite{sun2025formatron} which mask a model's vocabulary so that every emitted prefix extends to output in a regular or context-free language, giving a structural guarantee of \emph{completability}. The constraints they express are nonetheless weak relative to how models actually fail.

Syntax-only constraints do not address consequential program-generation failures such as scope, type compatibility, and declaration effects, none of which a grammar mask can express. Lifting the guarantee to semantic constraints gives two distinct properties:
\begin{enumerate}
    \item \emph{safe pruning}: no prefix with a semantically valid completion is rejected
    \item \emph{dead-end freedom}: every retained prefix still has such a completion
\end{enumerate}

The first protects the model's valid continuations while the second protects the decoding loop from being walked into a state where no token is admissible. They are different theorems with different hypotheses. 

We structure this paper around the following question:
\begin{quote}
	\emph{When may a semantic prefix oracle reject an LLM continuation without hiding a valid program, and how do we validate that the implementation actually respects that contract?}
\end{quote}

As a working example, take the finite lambda instance of \S\ref{sec:instances}. Under its constraints, a model that has decoded \texttt{let x : Int = } may still emit a digit, an open parenthesis, or any \texttt{Int} in scope, but the token \texttt{true} is masked out: the annotation has fixed the demanded type and the boolean literal's type is fixed too. This is where semantics is needed, since \texttt{let x : Int = true;} is a perfectly good parse in the underlying context-free grammar. 

Our declarative format attaches typing rules to a context-free surface. The rules compile to a language-parametric IR over first-order terms, pattern ascription, syntactic unification, and finite right-bound context effects. \aufbau executes this IR during Earley descent and rejects a path only after a stable semantic contradiction. The kernel has no built-in notion of type, function, or arrow. Those meanings come from each grammar's constructors and shared holes. Safe pruning is therefore a property of the oracle. Dead-end freedom is a stronger property of particular grammars. In \S\ref{sec:soundness} we give three sufficient conditions, show that the ML and C-like instances satisfy them, and show why the STLC instance used in our experiments does not.

Because a proof of the rejection rule does not rule out implementation bugs, we also validate Aufbau \emph{differentially} against production compilers sharing none of our code. \pseven then measures what the semantic layer buys over the same mask with typing removed.

This paper contributes:
\begin{itemize}
    \item A semantic-prefix formalism that separates safe pruning from dead-end freedom, with a safe-pruning theorem, simple sufficient conditions for dead-end freedom, and a tokenizer-lifting lemma under an explicit vocabulary-coverage hypothesis (\S\ref{sec:model}, \S\ref{sec:tool})
	\item A language-parametric implementation (\aufbau/\pseven, with STLC, finite-lambda, core ML, C-like, and typed tool-DSL instances in three surface dialects) and prefix-by-prefix differential validation against \texttt{ocamlc} and \texttt{cc}: zero false prunes over every prefix of 65 in-fragment valid programs, 42/42 on a recursion probe, mid-stream error localization at 25/30 vs.\ 0/30 for a syntax-only oracle (\S\ref{sec:tool}--\S\ref{sec:cert})
	\item A twelve-model generation study containing a matched semantic-versus-syntactic ablation for the nine models that ran both arms, 3{,}693 graded attempts in total, holding parser, sampler, tokenizer bridge, and surface grammar fixed, with per-language validity and correctness metrics reported separately, together with exploratory observations on reasoning-then-constrained scheduling, rejection--uncertainty association, and multi-turn tool episodes (\S\ref{sec:gen}).
\end{itemize}

The artifact contains the implementation, specifications, data, and commands used for the reported measurements.

\section{The Constraint Model}\label{sec:model}

\subsection{Grammar Specifications}\label{sec:spg}
We fix an alphabet $\Sigma$ and write $L \subseteq \Sigma^\ast$ for a language. A prefix $s$ is \emph{completable} for $L$ when $s s' \in L$ for some $s'$. Prefix parsing accepts exactly the completable prefixes of the surface syntax. The surface parser keeps branches compatible with the consumed input, and the semantic layer removes branches carrying a stable contradiction. Writing $R_G$ for the prefixes retained by Aufbau, the guarantee proved below (\S\ref{sec:soundness}) is the one-sided inclusion
\[
	\operatorname{Pref}(L_G) \subseteq R_G,
	\qquad
	\operatorname{Pref}(L_G) = \{s \mid \exists r.\; sr \in L_G\}.
\]

\begin{definition}[Semantic grammar specification]\label{def:spg}\setlength{\leftmargin}{0pt}\setlength{\rightmargin}{0pt}
 A \emph{semantic grammar specification} (what an .auf file can represent) is a tuple $G = (N, T, P, S, \Theta, \mathcal{T}, \mathcal{B})$: nonterminals $N$, terminal recognizers $T$, productions $P$, start symbol $S$, semantic rules $\Theta$, a partial map $\mathcal{T} : N \rightharpoonup \Theta$ attaching a rule to a nonterminal, and binding identifiers $\mathcal{B}$. A production for $n$ is a sequence $\alpha_1[b_1]\cdots\alpha_m[b_m]$ where each $\alpha_i \in T \cup N$ and each $b_i \in \mathcal{B} \cup \{\varepsilon\}$ names the child's evidence for use by a rule.
\end{definition}

We reserve \emph{semantic prefix grammar} for the subclass additionally satisfying the three conditions of \S\ref{sec:soundness} that guarantee witness construction. Bindings are annotations only, and do not affect the context-free structure. Terminals are matched through a segment interface returning a status in $\{\mathsf{Prefix},$ $\mathsf{Exact},$ $\mathsf{Extensible}\}$ together with a regular residual, the derivative of the recognizer after the consumed segment~\cite{brzozowski1964derivatives,owens2009regex}. The minimal example, used throughout:
{\small
\[
	\begin{aligned}
		\text{Id}   & \to \mathcal{R}(\texttt{[A-Za-z][A-Za-z0-9]}^\ast)                                                                    \\
		\text{Type} & \to \text{Id} \mid \text{Type}\;\text{`$\to$'}\;\text{Type} \mid \text{`('}\,\text{Type}\,\text{`)'}                  \\
		\text{Var}  & \to \text{Id}[x]                                                                                     & (\textsc{var}) \\
		\text{Lam}  & \to \text{`$\lambda$'}\,\text{Id}[a]\,\text{`:'}\,\text{Type}[t]\,\text{`.'}\,\text{Expr}[e]         & (\textsc{lam}) \\
		\text{App}  & \to \text{Expr}[l]\;\text{Expr}[r]                                                                   & (\textsc{app}) \\
		\text{Expr} & \to \text{Var} \mid \text{Lam} \mid \text{App}
	\end{aligned}
\]}
with one rule per annotated nonterminal, each a set of \emph{ascriptions} over shared holes $\hole{A}, \hole{B}$:
{\small
\[
	\frac{x : \hole{A} \ \text{in}\ \Gamma}{\,:\,\hole{A}}\,(\textsc{var})
	\;\;
	\frac{\Gamma[a{:}t] \vdash e : \hole{B}}{\,:\,t \to \hole{B}}\,(\textsc{lam})
	\;\;
	\frac{l : \hole{A} {\to} \hole{B} \quad r : \hole{A}}{\,:\,\hole{B}}\,(\textsc{app})
\]}

The \textsc{app} rule has no idea what a function is. Its premise $l : {\hole{A} \to \hole{B}}$ is one constructor pattern: $\to$ is a constructor \emph{because the grammar's type fragment wrote it}, and $\hole{A}, \hole{B}$ bind its subterms. There is no separate algebra of types, leaves are regular languages, and the arrow means only what hole-sharing across the three rules makes it mean. The inference notation is a front end to a small IR of term evaluation, ascription, equality, context membership, scoped extension, and right-bound effects, and the same IR runs the ML, C, and tool instances with no distinguished type or function instruction.

\subsection{Evidence and Constraint Graphs}\label{sec:evidence}
To keep the mechanism language-parametric, semantic evidence is represented only by terms over the grammar's signature, with \emph{ascription} for checking a term against an expected pattern and first-order \emph{unification} as the sole rewriting mechanism~\cite{robinson1965,martelli1982}. Context lookup, type equality, and structural matching are all ascriptions discharged by unification. Equality $\tau_1 = \tau_2$ is $\tau_1 : \tau_2$. Membership $x \in \Gamma$ is $x : \hole{A}$ solved against $\Gamma$. A shape such as $\hole{A} \to \hole{B}$ is matched constructor-against-constructor with each hole bound by position, with no decomposition step and no built-in notion of domain or codomain.

Parsing a prefix yields an AST whose nodes carry, at bound positions, \emph{evidence}: a status in $\{\mathsf{Exact}, \mathsf{Prefix}, \mathsf{Extensible}\}$ propagated bottom-up, a type term possibly containing holes, an exported context effect, and a binding map sending each binding either to $\mathsf{Resolved}(\nu)$ (the child's evidence, once the dot has passed it) or $\mathsf{Pending}$ (the child is past the right frontier). A binding moves from $\mathsf{Pending}$ to $\mathsf{Resolved}$, never the reverse. This is the available/unavailable split of incremental attribute grammars~\cite{demers1981incremental} and the typed-hole reading of an unfilled position~\cite{omar2017hazelnut}.

Rules relate evidence directly, so semantic information forms a graph over evidence rather than a tree. Firing a rule emits one unification edge per premise plus a conclusion edge, and holes shared between ascriptions become equality edges. This is the constraint-generation reading of type inference where a typing system is a set of unification constraints over a partial program.
A vertex is \emph{frontier} when it still carries an incomplete-input signal (status $\mathsf{Prefix}$/$\mathsf{Extensible}$, or a $\mathsf{Pending}$ binding). The graph at prefix $s$, written $\mathcal{G}(s)$, is \emph{open} when some edge touches a frontier vertex and \emph{closed} when it is not open and the solver finds no contradiction. The parser builds and solves this graph during Earley descent rather than after parsing. A stable contradiction removes that item immediately, while unresolved frontier evidence remains in the chart. Context effects are right-bound. Only an $\mathsf{Exact}$ node exports a declaration into the context its right siblings see.

\begin{proposition}[Fixed-prefix decidability]\label{prop:decidable}
	Let $G$ be finite with regular terminals, and let every semantic rule compile to a finite straight-line IR program interpreted in the free first-order theory with the occurs check and no declared rewrites. Compute the nullable nonterminals by least fixed point and define the zero-consumption left-corner relation $A \to_0 B$: some production of $A$ can reach an occurrence of $B$ after a prefix of nullable symbols. Suppose the $\to_0$ graph restricted to nonterminals reachable from $S$ is acyclic, and, separately, that every reachable nonterminal is productive (derives some finite terminal string). Then evaluation of any finite prefix and masking over any finite candidate set are decidable.
\end{proposition}
\noindent\emph{Proof (sketch).} Earley parsing over a finite input has a finite chart of item cores. Semantic elaboration can multiply items only by re-deriving the same core with different evidence. Acyclicity of $\to_0$ bounds the number of rule firings attributable to any span, so each core carries finitely many semantic refinements and context effects. Terms are finite, regular-leaf intersection is decidable, and first-order unification terminates with the occurs check. Semantic pruning only removes annotated items, so the annotated chart reaches a fixed point, and the same procedure applied to each member of a finite candidate set terminates. \qed

Zero-consumption is needed, because a unit cycle such as $A \to B$, $B \to A \mid \texttt{x}$ consumes no input while its rules wrap evidence on every traversal, giving one finite span unboundedly many annotations. Productivity covers a separate failure mode, since a consuming cycle can still be unproductive, as $S \to \texttt{x}\,S$ is. Both are decidable by least-fixed-point computation and are checked when a grammar is loaded.
The proposition decides whether a \emph{fixed} finite continuation survives the semantic Earley chart. Whether an arbitrary prefix has \emph{some} valid suffix is a different question, and needs the analysis below. Declared normalization rewrites are outside Proposition~\ref{prop:decidable}, since they need separate termination and confluence arguments.

\subsection{Verdicts: Safe Pruning and Dead-End Freedom}\label{sec:soundness}
\begin{definition}[Witness, denotation, verdict]\label{def:eval}
    A \emph{witness} for a prefix $s$ is $r \in \Sigma^\ast$ such that $sr$ admits a complete parse of the start symbol whose graph $\mathcal{G}(sr)$ is closed and consistent. The \emph{denotation} $\mathcal{D}(s)$ is the set of witnesses. The semantically valid language $L_G$ is the set of complete inputs with an empty-witness parse. The evaluator returns
    \[
        \evalf(\mathcal{G}(s)) =
        \begin{cases}
            \Satisfied & \text{closed and consistent}, \\
            \Lost      & \text{stable contradiction}, \\
            \Live      & \text{otherwise}.
        \end{cases}
    \]
    A contradiction is \emph{stable} when later input cannot repair it. This includes rigid constructor clashes, regular-leaf clashes, occurs-check failures, and failed lookups in the fixed left context. Ambiguity is existential: a prefix is \Lost{} only when every compatible branch is \Lost.
\end{definition}

\begin{theorem}[Safe pruning]\label{lem:safe}
    Aufbau never rejects a prefix that has a witness.
\end{theorem}
\begin{proof}
    Every rejection is justified by a stable contradiction on every compatible branch. Extending a branch can instantiate unresolved holes, refine regular residuals as more input is consumed, and append right-bound context effects. The rejection forms above are monotone under those operations: a rigid constructor clash remains a clash, an empty regular intersection cannot be restored by further refinement, an occurs-check failure cannot be repaired by later substitution, and a failed lookup in the fixed left context cannot be repaired by declarations exported to the right. Thus a contradiction classified as stable persists on every extension of that branch. If $s$ had a witness $r$, the complete consistent parse of $sr$ would restrict to a compatible branch at $s$ that is not \Lost, contradicting rejection.
\end{proof}

The implementation is conservative about information that is not yet fixed. In particular, prediction-time checks use only constraints already determined by the consumed prefix. Missing information can therefore delay a rejection, but it cannot create one. This matters for dead-end freedom, not for safe pruning.

Dead-end freedom is a property of a grammar in addition to the oracle. A \Live{} prefix has not contradicted the grammar, but it may still have no completion. Plain STLC gives the simplest example. With distinct atomic types $A$ and $B$, the prefix $\lambda f{:}A{\to}B.\ f($ demands a term of type $A$. If the language has no closed term of type $A$, the prefix is \Live{} but has no witness.

For the typed grammars in this paper, the relevant notion of inhabitation is simply \emph{type coverage}: every type that can be demanded at a reachable prefix has some finite closed term of that type. A universal inhabitant is one way to obtain coverage, but it is not required. Throughout, a parser branch, nonterminal, regular residual, or type demand is \emph{reachable} when it occurs after consuming some finite input prefix from the start symbol on a compatible branch, before that branch becomes \Lost.

\begin{definition}[Dead-end-free grammar class]\label{def:inhabitation}
    We use the following sufficient conditions for dead-end freedom.
    \begin{enumerate}
        \item[(P)] \emph{Surface productivity.} Every reachable nonterminal has a finite completion, and every reachable regular residual is nonempty.
        \item[(T)] \emph{Type coverage.} Every reachable type demand has a finite closed realizer whose construction does not invalidate the already fixed left context.
        \item[(F)] \emph{Left-to-right constraint flow.} A child's semantic obligations are determined by the parent demand, the fixed context, and children to its left. Later children do not constrain earlier ones.
    \end{enumerate}
\end{definition}

\begin{theorem}[Dead-end freedom under coverage]\label{prop:complete}
    A grammar satisfying (P), (T), and (F) is dead-end free: every \Live{} prefix has a witness.
\end{theorem}
\begin{proof}[Proof sketch]
    Complete the leftmost unfinished position repeatedly. Condition (P) closes purely syntactic and regular positions. When a typed expression is required, (F) makes its demanded type depend only on information already fixed to its left, and (T) supplies a finite closed term of that type. Repeating this process moves strictly rightward and eventually closes the parse. The chosen terms satisfy the demands under which they were inserted, so the resulting completion is consistent.
\end{proof}

The conditions are deliberately sufficient rather than a characterization of every grammar on which Aufbau happens to avoid dead ends. They are useful because they separate the oracle property from the language property. Safe pruning holds independently. Dead-end freedom follows once a grammar has productive syntax, type coverage, and left-to-right semantic flow.

The examples make the boundary concrete, and the boundary runs through our own instances. Our STLC instance (\S\ref{sec:instances}) does not satisfy (T): its base-type universe is unrestricted, so an atomic type may have no closed inhabitant, and the $\lambda f{:}A{\to}B.\ f($ prefix above is a prefix of that grammar. Coverage is recoverable by construction: restrict the base-type universe to a finite set in which every base type has a literal. The literals cover the base types, and arrow types follow inductively, since a closed realizer $e_B$ of $B$ makes $(x{:}A) \Rightarrow e_B$ a closed realizer of $A \to B$. The finite-lambda instance below uses exactly this construction. We did not run the reported STLC experiments on that restricted instance. Core ML has a simpler cover because \texttt{assert false} can be assigned every expression type. The C-like fragment also covers every reachable expression-type demand in its fragment: scalar types have literals, and explicit casts from a scalar literal realize pointer types and the other concrete expression types admitted by the grammar. The tool DSL has no comparable universal cover. With an empty context, \texttt{let s = summarize(} demands a \texttt{docs} value, but the grammar has no closed \texttt{docs} literal, so the prefix can remain \Live{} with no completion.

Nothing in this argument requires the implementation to decide (P), (T), and (F) automatically. They state when the stronger dead-end-free guarantee applies. The implementation itself only relies on the safe-pruning contract.

\subsection{Worked Example}\label{sec:example}
Figure~\ref{fig:worktree} illustrates the three verdicts. The complete term resolves $x$ in the fixed context and closes the graph. A pending body leaves the graph open and therefore \Live, without guaranteeing a witness. The exact variable $y$ fails lookup in $\Gamma=[x{:}\mathtt{int}]$, yielding a stable contradiction that syntax alone cannot detect.

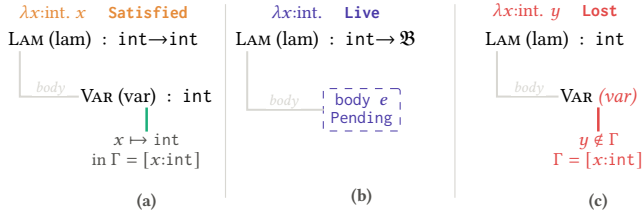
\begin{figure}[t]
	\centering
\footnotesize
\resizebox{\columnwidth}{!}{%
	\begin{tikzpicture}[
			every node/.style={font=\footnotesize},
			lvl/.style={align=left, inner sep=1pt},
			ev/.style={font=\ttfamily\scriptsize, text=ink, align=center, inner sep=1pt},
			pending/.style={font=\ttfamily\scriptsize, text=purple, align=center, draw=purple, dashed, rounded corners=1pt, inner xsep=3pt, inner ysep=2pt},
			rule/.style={font=\scriptsize\itshape, text=green},
			bad/.style={font=\scriptsize\itshape, text=red},
			astedge/.style={draw=grid, line width=0.8pt, line cap=rect, line join=miter},
			goodedge/.style={draw=green, line width=1pt, line cap=rect},
			badedge/.style={draw=red, line width=1pt, line cap=rect},
			branchlabel/.style={font=\tiny\itshape, text=grid, fill=white, inner xsep=1pt, inner ysep=0pt},
			separator/.style={draw=grid, line width=0.45pt},
			panelletter/.style={font=\scriptsize\bfseries, text=ink}
		]
		\newcommand{\ty}[1]{{\texttt{#1}}}
		\draw[separator] (17mm,5mm) -- (17mm,-22mm);
		\draw[separator] (49mm,5mm) -- (49mm,-22mm);
		\begin{scope}
			\node[font=\scriptsize\bfseries, text=orange] (aprog) at (0,4mm)
			{$\lambda x{:}\mathrm{int}.\;x$ \quad \Satisfied};
			\node[lvl] (a0) at (0,0)
			{\textsc{Lam} \textrm{(lam)} \;:\; \ty{int$\to$int}};
			\node[lvl, below=4.5mm of a0, xshift=6mm] (a1)
			{\textsc{Var} \textrm{(var)} \;:\; \ty{int}};
			\draw[astedge]
			($(a0.south west)+(2mm,0)$)
			|- node[branchlabel, pos=0.76, above] {body}
			(a1.west);
			\node[ev, below=3mm of a1] (a2)
			{$x \mapsto \ty{int}$\\
			{\normalfont in $\Gamma=[x{:}\mathtt{int}]$}};
			\draw[goodedge] (a1.south) -- (a2.north);
			\node[panelletter, below=2mm of a2] {(a)};
		\end{scope}
		\begin{scope}[xshift=31mm]
			\node[font=\scriptsize\bfseries, text=purple] (bprog) at (0,4mm)
			{$\lambda x{:}\mathrm{int}.$ \quad \Live};
			\node[lvl] (b0) at (0,0)
			{\textsc{Lam} \textrm{(lam)} \;:\; \ty{int$\to$\,\hole{B}}};
			\node[pending, below=5mm of b0, xshift=5mm] (b1)
			{body $e$\\[-1pt] \textsc{Pending}};
			\draw[astedge]
			($(b0.south west)+(2mm,0)$)
			|- node[branchlabel, pos=0.74, above] {body}
			(b1.west);
			\node[panelletter, below=7mm of b1] {(b)};
		\end{scope}
		\begin{scope}[xshift=63mm]
			\node[font=\scriptsize\bfseries, text=red] (cprog) at (0,4mm)
			{$\lambda x{:}\mathrm{int}.\;y$ \quad \Lost};
			\node[lvl] (c0) at (0,0)
			{\textsc{Lam} \textrm{(lam)} \;:\; \ty{int}};
			\node[lvl, below=4.5mm of c0, xshift=6mm] (c1)
			{\textsc{Var} {\itshape\color{red}(var)}};
			\draw[astedge]
			($(c0.south west)+(2mm,0)$)
			|- node[branchlabel, pos=0.76, above] {body}
			(c1.west);
			\node[ev, text=red, below=3mm of c1] (c2)
			{$y \notin \Gamma$\\
			{\normalfont $\Gamma=[x{:}\mathtt{int}]$}};
			\draw[badedge] (c1.south) -- (c2.north);
			\node[panelletter, below=2mm of c2] {(c)};
		\end{scope}
	\end{tikzpicture}%
}
	\caption{The oracle's evidence on three prefixes of the minimal grammar. Gray
	right-angle links show AST structure, colored links semantic evidence.
	The complete term closes at $\mathtt{int}\to\mathtt{int}$. The missing
	body stays \textsc{Pending}. The exact variable $y$ fails lookup in
	$\Gamma=[x{:}\mathtt{int}]$, making the last prefix \Lost.}
	\Description{Three evidence graphs. The complete term lambda x colon int dot x is Satisfied with type int to int, lambda x colon int with a pending body is Live and lambda x colon int dot y is Lost because y is absent from the context.}
	\label{fig:worktree}
\end{figure}
\section{\aufbau, \ptool, and the Instances}\label{sec:tool}

\subsection{Architecture}\label{sec:arch}

A language definition is lowered into Aufbau's semantic grammar IR, productions, regular terminals, binding positions, compiled rule programs, context effects, with \texttt{.auf} as one readable front end and the API able to construct the same IR directly. \aufbau executes it inside an Earley parser~\cite{earley1970efficient} whose items carry the evidence and constraint graph of \S\ref{sec:evidence}.

 The entry point \texttt{try\_feed} extends the current prefix by a candidate string and returns the three-valued verdict. \pseven puts this oracle in the decoding loop (Figure~\ref{fig:loop}).

\begin{figure}[t]
	\centering
	\begin{tikzpicture}[
		system/.style={
				draw,
				rounded corners=6pt,
				minimum width=1.75cm,
				minimum height=0.78cm,
				align=center,
				font=\scriptsize\bfseries,
				line width=0.8pt
			},
		outcome/.style={
				draw,
				rounded corners=4pt,
				minimum width=2.05cm,
				minimum height=0.72cm,
				align=center,
				font=\tiny,
				line width=0.7pt,
				inner xsep=3pt
			},
		flow/.style={
				->,
				line width=0.85pt,
				draw=ink!70
			},
		label/.style={
				font=\tiny,
				fill=white,
				inner sep=1pt
			}
	]

	\node[
		system,
		fill=blue!10,
		draw=blue!70!black
	] (llm) {
		LLM\\model
	};

	\node[
		system,
		fill=green!10,
		draw=green!70!black,
		right=1.6cm of llm
	] (p7) {
		P7\\controller
	};

	\node[
		system,
		fill=orange!11,
		draw=orange!75!black,
		right=1.4cm of p7
	] (oracle) {
		Aufbau\\oracle
	};

	\draw[flow]
	(llm.east)
	--
	node[label] {token}
	(p7.west);

	\draw[flow]
	([yshift=-0.12cm]p7.north east)
	-- node[label] {\texttt{try\_feed}}
	([yshift=-0.12cm]oracle.north west);

	\draw[flow, <-]
	([yshift=0.12cm]p7.south east)
	-- node[label] {verdict}
	([yshift=0.12cm]oracle.south west);

	\node[
		outcome,
		fill=red!9,
		draw=red!72!black,
		below left=1.3cm and 0.2cm of p7
	] (lost) {
		\textbf{\textsc{Lost}}\\
		exclude and resample
	};

	\node[
		outcome,
		fill=orange!9,
		draw=orange!72!black,
		below=1.3cm of p7
	] (live) {
		\textbf{\textsc{Live}}\\
		commit and continue
	};

	\node[
		outcome,
		fill=green!11,
		draw=green!72!black,
		below right=1.3cm and 0.2cm of p7
	] (satisfied) {
		\textbf{\textsc{Satisfied}}\\
		commit and finish
	};

	\draw[flow, draw=red!70!black]
	(p7.south west)
	--
	(lost.north);

	\draw[flow, draw=orange!75!black]
	(p7.south)
	--
	(live.north);

	\draw[flow, draw=green!70!black]
	(p7.south east)
	--
	(satisfied.north);

	\draw[
		flow,
		draw=red!62!black
	]
	(lost.west) -- ++(-0.20,0)
	-- (llm.south)
	node[label, pos=0.4] {resample};

\end{tikzpicture}
	\caption{\textbf{P7 constrained-decoding architecture.} The model proposes a
		token, P7 submits it to the Aufbau prefix oracle, and the oracle
		returns one of three verdicts: \textsc{Lost} is excluded and
		resampled, \textsc{Live} is committed and decoding continues,
		\textsc{Satisfied} completes the derivation.}
	\Description{An LLM sends a proposed token to P7, which calls the Aufbau oracle. Lost excludes and resamples the token, Live commits it and continues, and Satisfied commits it and finishes.}
	\label{fig:loop}
\end{figure}
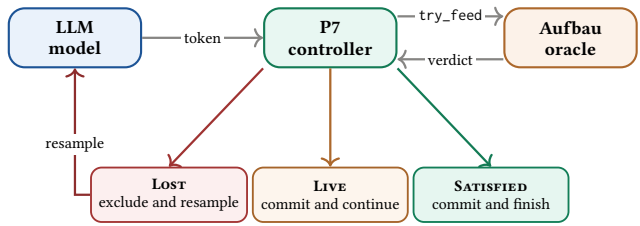

 Two properties of the loop matter below. The mask never forces completion: the model's own stop tokens stay admissible mid-derivation, so \S\ref{sec:gen} counts abandonments rather than hiding them. The mask prunes rather than repairs: given a model with no useful solution, it can turn free-form garbage into well-formed garbage that still fails grading.

\paragraph{Tokenizer lifting.}

\pseven feeds \aufbau tokens, not characters, so the guarantees above must cross that boundary. Write $\Sigma_G = \bigcup_{R \in T} \operatorname{alph}(R)$ for the union of the recognizers' alphabets, $\ell : V \to \Sigma_G^\ast$ for the exact spelling map on ordinary text tokens of the vocabulary $V$ (control tokens and \textsc{eos} follow the completion policy above), and $L_V = \{\ell(v) \mid v \in V\}$.

 The \emph{coverage hypothesis} $\Sigma_G \subseteq L_V^\ast$ gives $\Sigma_G^\ast \subseteq L_V^\ast$, since $L_V^\ast$ is closed under concatenation.
\begin{lemma}[Tokenizer lifting]\label{lem:tok}
	Under exact spelling (the oracle consumes, byte-for-byte, the text the model conditions on) and coverage, every character-level witness factors into a finite token sequence, and regular derivatives carry the oracle's state across a token ending inside a terminal recognizer. Character-level witness existence therefore lifts to token-sequence reachability.
\end{lemma}
\begin{proof}
	Coverage factors the witness into a concatenation of spellings $\ell(v_1)\cdots\ell(v_k)$, exactly the string \pseven submits token by token, and composing the recognizer's derivative across $v_1,\ldots,v_k$ reaches the residual the character-level run would.
\end{proof}

\subsection{Language Instances}\label{sec:instances}

\paragraph{STLC}
A monomorphic simply-typed lambda calculus with mandatory binder annotations, used both by the oracle's validation harnesses and as the STLC arm of the generation study. Its base-type universe is unrestricted, so this instance does not satisfy the type-coverage condition (T) of \S\ref{sec:soundness} and we make no dead-end-freedom claim for it: the $\lambda f{:}A{\to}B.\ f($ counterexample is a prefix of this grammar. Safe pruning applies to it exactly as to every other instance.

\paragraph{Finite Lambda (\texttt{fun.auf}).}
A variant of the lambda calculus with modern surface syntax, \texttt{let}-bindings, and a base-type universe restricted to \texttt{Int}, \texttt{Bool}, and \texttt{Float}, each equipped with literals. These literals cover the base types. Function types are covered inductively: if $e_B$ is a closed realizer of $B$, then \texttt{(x: A) => }$e_B$ is a closed realizer of $A \to B$. Together with productive syntax and left-to-right constraint flow, the instance therefore satisfies (P), (T), and (F), unlike the unrestricted STLC instance used in the experiments.

\paragraph{Core ML (\texttt{ml.auf}).}
A monomorphic ML fragment with \texttt{let}/\texttt{let rec}, functions, application, integers, booleans, comparison and arithmetic, lists (\texttt{T list}, cons, \texttt{[]}), \texttt{if}/\texttt{then}/\texttt{else}, and \texttt{assert false}. The concrete syntax is a strict OCaml subset, so every accepted program feeds \texttt{ocamlc} with no translation. This makes differential prefix validation possible (\S\ref{sec:cert}). The axiom $\texttt{assert false} : \hole{A}$ is a universal expression inhabitant, so every reachable ML expression-type demand is covered. Checking-mode \texttt{let rec} is incremental because its signature is fixed before the body streams.

\paragraph{A C-like fragment (\texttt{c.auf}).}
A typed C-like fragment whose accepted output is graded by \texttt{cc -fsyntax-only}: declarations and assignment, typed functions and calls, \texttt{if}, \texttt{while}, and \texttt{for}, arithmetic and comparison over \texttt{int}, \texttt{float}, and \texttt{char}, pointers with address-of and dereference, and explicit casts.

The fragment threads each function's declared return type through an ambient context entry and checks every \texttt{return} against it, admitting a bare \texttt{return;} only in a \texttt{void} function.\footnote{The return-type check postdates the generation study of \S\ref{sec:gen}, which ran on the earlier version of \texttt{c.auf} without it. The reported C generation numbers are therefore for that earlier grammar and rerunning them is future work. \S\ref{sec:gen} repeats this caveat where the numbers appear.} We still call it C-like deliberately: it does not prove that every path through a non-\texttt{void} function returns, does not constrain the operand of a cast or of address-of, and its demanded-type universe includes \texttt{void} and unbounded pointer depth, so Aufbau-accepted programs are not guaranteed to be accepted, or accepted without diagnostics, by a C compiler. 

For dead-end freedom, the relevant point is narrower than full C validity. Every reachable expression-type demand in this fragment has a closed realization: scalar types have literals, while explicit casts from a scalar literal realize the pointer types admitted by the grammar. Together with the fragment's productive statement forms and left-to-right declaration flow, this gives the coverage used in Theorem~\ref{prop:complete}. The differential harness of \S\ref{sec:cert} separately checks compatibility with \texttt{cc}.

Two design points carry the semantics. A two-sort type discipline keeps the concrete sort \texttt{CType} (scalars and pointers, what a cast or declaration can spell) distinct from the internal signature sort that types functions. The latter is unreachable from concrete syntax, so no source expression can demand a function-shaped type. And definition-before-use holds at translation-unit granularity: a call is checked against a signature already exported by an earlier definition's right-bound effect, which a left-to-right oracle handles without lookahead because the definition precedes the call in the token stream.

\paragraph{Declarative boundary: self-recursion.}
Inference-mode recursion cannot be declared in this fragment. \texttt{let rec} without an annotation was implemented and dropped after some empirical failure evidence. The problem is that the recursive call races the still-streaming parameter list to resolve the signature metavariable, and an early unification can commit it to a shape the header later contradicts, producing a false prune rather than a real error. Checking-mode recursion, whose signature is exact before the body opens, has no such race, so the declarable fragment admits checking-mode program recursion only and inference-mode program recursion is excluded as a restriction of the fragment. It could theoretically be implemented but would require dropping the level of constraints which isn't very interesting here.

This important limitation was found by our test system probing \texttt{ocamlc} against the oracle. Indeed, some experiments on recursion corners surfaced the discrepancy (\S\ref{sec:cert}). The instrumentation that rejects the model's tokens is what caught the grammar author. This establishes the importance of our differential certification approach for implementation correctness.

\paragraph{The agentic \texttt{tool} DSL dialects.}

The third instance exercises the semantic layer as an agent harness. Constrained decoding is usually evaluated on single-shot synthesis, but multi-turn tool use, where the typed context grows as the episode runs, is the natural application for a typed oracle. \texttt{tool.auf} is an invented typed pipeline language over a fixed registry (\texttt{search}: string$\to$docs, \texttt{summarize}: docs$\to$string, \texttt{count}: docs$\to$int, \texttt{format}: int$\to$string) with \texttt{let}-sequencing and \texttt{return}.

Being invented, it has no pretraining prior and no external compiler, so an episode is graded by the \emph{value} its \texttt{return} evaluates to against the registry's real implementations rather than by well-typedness. Each turn is one constrained step. The binding it introduces is pushed into $\Gamma$, so turn $n{+}1$ decodes under the types established by turns $1..n$, threading the constraint graph of \S\ref{sec:evidence} across the episode.

Three \emph{dialects} share one type discipline: a C-like \texttt{let} syntax (\texttt{tool.auf}), S-expressions (\texttt{tool\_sexpr.auf}), and XML elements (\texttt{tool\_xml.auf}). Only the surface varies, so the dialect axis separates semantic discipline from syntax priors and tokenizer alignment.

The dialects also show why registry reachability is weaker than type coverage. With an empty context, \texttt{let s = summarize(} is \Live{} with no completion because the demanded \texttt{docs} argument has no literal form and no binding supplies one. We therefore make no dead-end-freedom claim for the tool DSLs. A 51-case harness checks the checker given a witness by injecting a binding of the demanded type at each demanding argument position and confirming acceptance.

\section{Empirical Findings}\label{sec:eval}
The evidence comes in two tiers with different epistemic status, and we keep them apart.

\paragraph*{Tier 1 (\S\ref{sec:cert})} tests the \emph{oracle}: CPU-only, deterministic, regenerable byte-for-byte by a reviewer with no GPU.
\paragraph*{Tier 2 (\S\ref{sec:gen})} measures what the oracle buys in \emph{generation}: GPU runs over open models, single greedy attempts, graded by per-language metrics that are deliberately not pooled into one accuracy (\S\ref{sec:gen} defines each).

\subsection{Differential Prefix Validation}\label{sec:cert}
The oracle's verdicts are compared prefix by prefix against a production compiler that shares none of our code: \texttt{ocamlc} for the ML instance and \texttt{cc} for the C instance. These are finite, hand-selected test cohorts, so the results validate rather than certify. The harness extends mechanically.

Safe pruning (Theorem~\ref{lem:safe}) applies to every instance. The cohorts below add finite external evidence for the ML and C implementations, the two instances with production-compiler comparators. Differential testing can expose implementation bugs, but finite cohorts are not an exhaustive correctness proof. Theorem~\ref{lem:safe} is in any case a paper proof of the rejection discipline, and the shipped engine is an optimized realization of it: these cohorts check selected end-to-end compatibility properties rather than establishing that the optimized engine computes the theorem's verdicts at every token boundary. A slow reference evaluator built directly from the formal rules, enabling prefix-by-prefix correspondence testing against the shipped engine, is future work.

\emph{False prunes.} Over the 65 in-fragment valid programs of the cohort (40 ML, 25 C), the oracle was evaluated at \emph{every} prefix and never reported \Lost: $0/40$ on ML and $0/25$ on C, with Wilson 95\% upper bounds of 8.8\% and 13.3\% (5.6\% pooled). The bound is honest about cohort size, and extension is mechanical: add programs, rerun, diff every prefix.

\emph{The recursion probe.} 42 programs probing \texttt{let rec} corners agree with \texttt{ocamlc} in 42/42 cases (Wilson 95\% CI $[0.92, 1.00]$). The probe exists because recursion is where our fragment boundary was wrong once (\S\ref{sec:instances}): the harness that found the unsoundness now pins the repaired boundary.

\emph{Prune lead time.} On 30 invalid programs, the semantic oracle reports \Lost{} mid-stream in 25 cases (12/17 on ML, 13/13 on C), where a syntax-only oracle catches 0 before end of input. In the caught cases the \Lost{} position trails the start of the compiler's diagnostic span by 2.3 characters pooled (2.5 on ML, 2.1 on C).

\emph{Oracle cost.} Median whole-input \tryfeed{} time on four synthetic stress families grows smoothly and mildly superlinearly with expression size, at log--log slopes of roughly 1.1--1.4 (Figure~\ref{fig:runtime}). The largest inputs, 450-token C declaration sequences and 65-application STLC chains, the adversarial case for chart size, stay under ${\sim}130$~ms.

  The measurement is model-free and CPU-only, on a six-core laptop CPU.

\begin{figure}[t]
	\centering
\begin{tikzpicture}
\pgfplotstableread[col sep=comma]{
tokens,median_ms
1,0.3710
3,0.6510
5,0.9820
7,1.2690
9,1.5980
11,1.9880
13,2.3180
15,2.6290
17,2.9780
19,3.5150
21,3.8450
23,4.2380
25,4.6150
27,5.0420
29,5.4480
31,5.8340
33,6.5890
35,7.0180
37,7.5560
39,7.9000
41,8.4420
43,8.8440
45,9.5670
47,10.0510
49,10.4450
51,11.0200
53,11.6530
55,12.5490
57,13.2970
59,14.0050
61,14.4350
63,14.9620
65,15.5980
67,16.5550
69,16.7930
71,17.8090
73,18.7760
75,18.8270
77,19.7630
79,21.0920
81,21.6820
83,22.3340
85,22.7030
87,24.9000
89,25.3580
91,25.5780
93,26.1060
95,27.0860
97,27.5440
99,28.6590
101,29.2690
103,29.9040
105,30.9290
107,31.7690
109,32.4840
111,34.5750
113,34.6790
115,35.8790
117,36.9560
119,37.4240
121,38.9450
123,39.2700
125,40.1240
127,41.8910
129,42.7380
131,43.1930
}\rtdataA
\pgfplotstableread[col sep=comma]{
tokens,median_ms
3,0.6920
5,1.2640
7,1.1010
9,1.2740
11,1.4730
13,1.7280
15,1.9330
17,2.1510
19,2.3570
21,2.5830
23,2.8060
25,3.1870
27,3.4460
29,3.6530
31,3.8940
33,4.2190
35,4.4890
37,4.7110
39,4.9750
41,5.2940
43,5.5270
45,5.8150
47,6.1580
49,6.6350
51,6.9450
53,7.3710
55,7.8500
57,8.0300
59,8.4180
61,9.1090
63,9.2910
65,9.7790
67,9.9770
69,10.4940
71,10.8730
73,11.5710
75,11.7290
77,12.0470
79,12.5810
81,12.9490
83,13.8920
85,14.4130
87,14.4420
89,15.0870
91,15.5170
93,16.1070
95,16.5410
97,17.2980
99,17.9950
101,19.0140
103,19.7670
105,20.4200
107,20.5600
109,19.9520
111,20.7710
113,21.5740
115,22.5940
117,22.9630
119,23.9070
121,24.8610
123,25.1070
125,26.3730
127,29.7030
129,29.6900
}\rtdataB
\pgfplotstableread[col sep=comma]{
tokens,median_ms
9,1.7180
16,2.8930
23,4.1710
30,5.2850
37,6.4680
44,7.7710
51,9.0470
58,10.1710
65,10.8920
72,11.7800
79,12.8800
86,14.4940
93,15.6190
100,17.6570
107,18.5480
114,19.8420
121,21.4160
128,22.5560
135,23.5880
142,26.2460
149,27.6660
156,29.3320
163,31.7690
170,33.0900
177,34.5590
184,36.1490
191,37.8530
198,39.5010
205,41.2930
212,42.9840
219,44.4880
226,46.0760
233,47.9590
240,49.9970
247,51.5640
254,55.8880
261,58.7660
268,58.6310
275,60.2880
282,61.0970
289,64.6100
296,67.3680
303,69.2440
310,70.0220
317,70.6230
324,73.8800
331,72.7410
338,74.8970
345,77.5070
352,79.3170
359,81.8130
366,86.2860
373,90.2680
380,88.0610
387,91.4530
394,95.5830
401,99.6680
408,102.2140
415,104.4880
422,107.7170
429,109.9130
436,109.6940
443,122.8990
450,125.1980
}\rtdataC
\pgfplotstableread[col sep=comma]{
tokens,median_ms
2,1.0090
3,1.3780
4,1.7560
5,2.2360
6,2.6580
7,3.0940
8,3.5880
9,4.1790
10,4.7530
11,5.4880
12,6.0900
13,6.9190
14,7.7970
15,8.6600
16,9.4390
17,10.3830
18,11.1810
19,12.3510
20,13.5430
21,14.1910
22,15.5130
23,16.5810
24,17.6480
25,20.3470
26,20.5820
27,22.0710
28,23.5530
29,25.1370
30,27.3350
31,28.2910
32,30.0420
33,31.8640
34,33.7250
35,35.8170
36,37.7150
37,40.3400
38,42.2820
39,44.5720
40,51.5210
41,53.3100
42,48.4960
43,49.8270
44,52.7830
45,55.0530
46,59.2240
47,62.4610
48,66.6940
49,69.1700
50,72.3860
51,75.5670
52,78.7820
53,82.6050
54,86.4750
55,89.0750
56,89.5640
57,92.8000
58,97.1970
59,106.6430
60,111.7420
61,109.3130
62,114.4850
63,118.6170
64,125.4900
65,131.3310
}\rtdataD
\begin{axis}[
  paperaxis,
  width=\columnwidth, height=4.6cm,
  xmode=log, ymode=log,
  xmin=1, xmax=520,
  ymin=.3, ymax=180,
  xlabel={Expression size (tokens)},
  ylabel={Median validation time (ms)},
  legend pos=north west,
  legend cell align=left,
  legend style={draw=none, fill=none, font=\scriptsize},
  log basis x=10,
  log basis y=10,
]

\addplot+[color=blue, very thick, mark=*, mark options={scale=.58, fill=white, line width=.45pt}, mark repeat=8]
  table[x=tokens, y=median_ms] {\rtdataA};
\addlegendentry{Arithmetic flat sum}

\addplot+[color=green, very thick, mark=square*, mark options={scale=.58, fill=white, line width=.45pt}, mark repeat=8]
  table[x=tokens, y=median_ms] {\rtdataB};
\addlegendentry{Arithmetic nested}

\addplot+[color=purple, very thick, mark=triangle*, mark options={scale=.58, fill=white, line width=.45pt}, mark repeat=8]
  table[x=tokens, y=median_ms] {\rtdataC};
\addlegendentry{C declarations}

\addplot+[color=orange, very thick, mark=diamond*, mark options={scale=.58, fill=white, line width=.45pt}, mark repeat=8]
  table[x=tokens, y=median_ms] {\rtdataD};
\addlegendentry{STLC application chain}
\end{axis}
\end{tikzpicture}
	\caption{Median validation time vs.\ expression size on four stress families (log--log). Growth is smooth and mildly superlinear (log--log slopes ${\approx}1.1$--$1.4$).}
	\Description{Log-log line chart of median validation time in milliseconds against expression size in tokens for four grammar families. All four curves are smooth and close to straight lines with slopes slightly above one. C declaration sequences reach 125 ms at 450 tokens and STLC application chains 131 ms at 65 tokens.}
	\label{fig:runtime}
\end{figure}
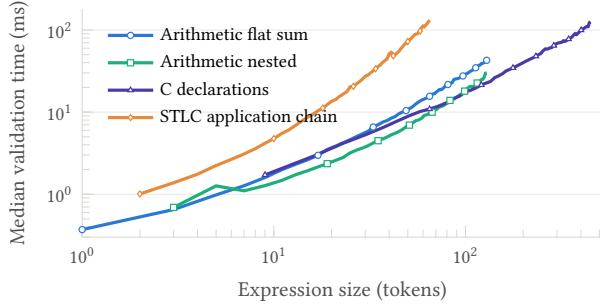

\subsection{Constrained Generation}\label{sec:gen}

\paragraph{Setup.}

We run twelve open models with full logit access: Qwen3.5 at 2B, 4B, and 9B in instruct and base variants, Qwen3.5-0.8B (instruct only), Qwen3.6-27B, Phi-4-mini, gpt-oss-20b, gemma-4-26b, and Pythia-1.4B~\cite{biderman2023pythia} as a no-instruction floor. The single-shot corpus is 52 tasks (33 STLC, 14 core ML, 5 C). The agentic corpus is tool episodes in three dialects. Each attempt runs in one of five decoding modes:

\begin{itemize}
	\item \textbf{Unconstrained}: free decoding, no mask.
	\item \textbf{Constrained Direct}: the semantic mask on every token.
	\item \textbf{Constrained Mixed}: free reasoning, then a constrained emission block, the interleaved schedule of CRANE~\cite{banerjee2025crane}.
	\item \textbf{Unconstrained Thinking}: a budgeted chain-of-thought phase, then a free answer~\cite{wei2022cot}.
	\item \textbf{Syntactic Only}: the same mask with the typing rules stripped, our within-framework ablation.
\end{itemize}

Three models did not run \textbf{Syntactic Only} (Qwen3.6-27B, gemma-4-26b, gpt-oss-20b), so pooled comparisons are matched over the models that ran both modes being compared. Decoding is greedy at temperature 0, one attempt, fixed seed. Intervals are Wilson 95\% and measure cohort uncertainty on hand-selected tasks, not population sampling.

Grading is per-language and deliberately heterogeneous, and we do not pool it into one accuracy:
\begin{itemize}
	\item STLC: \aufbau supplies the decoding oracle and an independent handwritten validator checks the output type and $\beta$/$\alpha$-equivalence against the reference.
	\item tool: the executed return value against the registry's real implementations. Behavioral.
	\item ML: an output counts as valid only when \aufbau accepts the declared type \emph{and} \texttt{ocamlc} independently accepts the program, and this program has the required type.
	\item C: an output counts as valid only when \aufbau and \texttt{cc -fsyntax-only} independently accept it.
\end{itemize}

The divergence idiom that gives ML its coverage also inhabits any demanded type at grading time: of 259 passing ML records, 36 contain \texttt{assert false} (9 of the 52 passing direct-mask records). \texttt{ocamlc}'s independent check on the declared type is what keeps the ML criterion independent of the oracle.

  ML and C passes are therefore compiler-validated \emph{validity} metrics rather than behavioral task-correctness metrics, and we reserve behavioral claims for STLC and tool. Each task is specified as a TOML file (prompt, initial prefix, reference program, resolution mode).

\paragraph{Per-effect decomposition.}

Gains are heterogeneous across model classes.

 Across the nine models that ran both ablation arms, every observed semantic-minus-syntactic model--language point estimate was nonnegative.  Constraint can repair weak and base models while taxing fluent ones: Qwen3.5-2B-Base rises from 0\% unconstrained to 50\% in mixed mode, while Qwen3.5-9B and Qwen3.6-27B both lose ground under the same mask.

Mixed mode avoids the cost by confining the mask to a late emission block, tracking full chain-of-thought within a few points while emitting 28--36 tokens instead of 1{,}100--2{,}000.

\paragraph{The semantic layer over syntax.}

The ablation isolates what the typing rules add over a pure CFG mask, holding parser, sampler, tokenizer bridge, and grammar fixed. This is the generation comparison the paper is about. The largest lifts by language are $+15.2$ points on STLC task correctness (Phi-4-mini) and $+14.3$ on ML validity (Qwen3.5-4B, Phi-4-mini, Qwen3.5-9B-Base), while every matched C cell is exactly 0.

 Aggregated within each language over the matched models, the semantic layer adds $+7.1$ points on ML (32.5\% against 25.4\%), $+2.0$ on STLC (40.4\% against 38.4\%), and exactly 0 on C (next paragraph). The ML figure carries the caveat that ML's validity grader shares the oracle's type system.

Syntax-only constraining can itself \emph{hurt}: pooled over models it sits below unconstrained on STLC ($38.4\%$ against $44.4\%$) and ML ($25.4\%$ against $30.2\%$), fighting the model's prior with no semantic information to spend. That replicates, inside our own stack, the degradation CRANE~\cite{banerjee2025crane} reports for constrain-from-the-first-token and the underperformance of grammar-only decoding in \textsc{ChopChop}'s tables~\cite{nagy2026chopchop}.

\begin{figure}[t]
	\centering
\begin{tikzpicture}
	\begin{axis}[
			paperaxis,
			width=\columnwidth, height=4.2cm,
			ybar, bar width=4pt,
			ymin=0, ymax=18,
			ylabel={semantic $-$ syntactic (pp)},
			symbolic x coords={Pythia,Qw-0.8B,Qw-2B,Qw-4B,Qw-9B,Phi-4,Qw-2B-B,Qw-4B-B,Qw-9B-B},
			xtick=data,
			xticklabel style={rotate=35, anchor=east, font=\tiny},
			legend columns=2,
			legend style={draw=none, fill=none, font=\scriptsize, at={(0.5,-0.32)}, anchor=north},
		]
		\addplot+[draw=none, fill=blue!55] coordinates
			{(Pythia,0) (Qw-0.8B,0) (Qw-2B,3.0) (Qw-4B,0) (Qw-9B,0) (Phi-4,15.2) (Qw-2B-B,0) (Qw-4B-B,0) (Qw-9B-B,0)};
		\addlegendentry{STLC}
		\addplot+[draw=none, fill=green!60] coordinates
			{(Pythia,7.1) (Qw-0.8B,0) (Qw-2B,0) (Qw-4B,14.3) (Qw-9B,7.1) (Phi-4,14.3) (Qw-2B-B,0) (Qw-4B-B,7.1) (Qw-9B-B,14.3)};
		\addlegendentry{ML}
	\end{axis}
\end{tikzpicture}
	\caption{Per-model, per-language semantic-minus-syntactic difference in percentage points. STLC and ML are shown separately. C is omitted because every matched C difference is zero.}
	\Description{Grouped bars show semantic-minus-syntactic pass-rate changes for nine models. All bars are nonnegative. The largest STLC increase is 15.2 points and the largest ML increases are 14.3 points. C is omitted because every matched change is zero.}
	\label{fig:heatmap}
\end{figure}
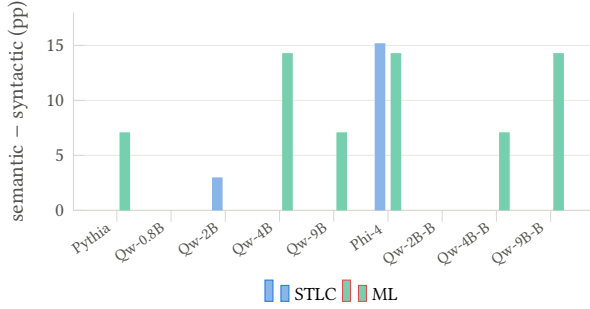

\paragraph{Where constraint helps, and a negative result.}

On the C-like fragment the direct mask lifts pooled validity from 53\% to 80\% ($+27$ points, matched over models that ran all five modes), but the syntax-only mask reaches the same 80\%. The ablation attributes none of the C gain to the typing rules, and the difference is not significant. Format control alone explains the lift, and we report that as a negative result rather than crediting definition-before-use or pointer typing. These C cells were produced before the return-type check of \S\ref{sec:instances} was added to \texttt{c.auf}, so they describe the earlier grammar. The added check tightens the semantic arm only, and rerunning the arm is future work. The semantic layer's language-level gain concentrates on ML instead ($+7.1$ pooled), where inference-style typing constrains many surface choices. On STLC the direct mask is slightly negative pooled (44\% to 40\%) while mixed mode reaches 55\%: those tasks are search-heavy rather than surface-heavy, so reasoning, not masking, binds.

 The tool dialects are the agentic case of \S\ref{sec:entropy}.

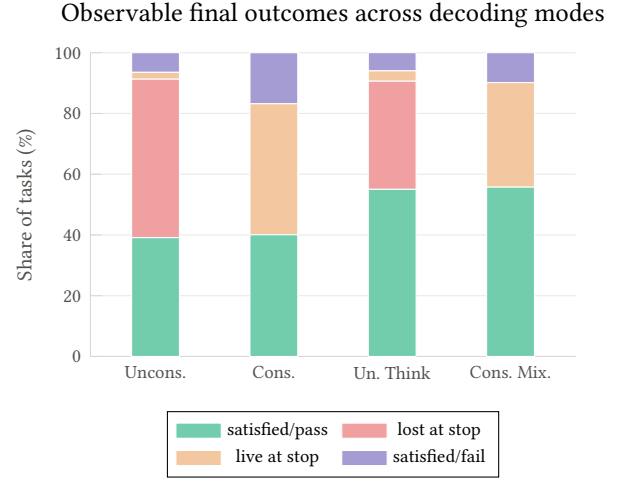
\begin{figure}[t]
	\centering
	\begin{tikzpicture}
	\pgfplotstableread[col sep=comma]{
		x,mode,SatisfiedPassed,LostAtStop,LiveAtStop,SatisfiedFailed
		1,plain,39.10,52.18,2.26,6.46
		2,direct,40.06,0.00,43.13,16.80
		3,thinking,55.02,35.60,3.40,5.99
		4,mixed,55.74,0.00,34.41,9.85
	}\errdata

	\begin{axis}[
			paperaxis,
			width=8cm, height=5.6cm,
			ybar stacked,
			bar width=18pt,
			ymin=0, ymax=100,
			xmin=0.45, xmax=4.55,
			ylabel={Share of tasks (\%)},
			xtick={1,2,3,4},
			xticklabels={Uncons.,Cons.,Un. Think,Cons. Mix.},
			x tick label style={align=center, font=\scriptsize},
			title={Observable final outcomes across decoding modes},
			legend columns=2,
			legend style={
					at={(.5,-.18)},
					anchor=north,
					font=\scriptsize,
					/tikz/every even column/.append style={column sep=3pt}
				},
			area legend,
		]
		\addplot+[draw=white, fill=green!62]
		table[x=x,y=SatisfiedPassed] {\errdata};
		\addlegendentry{satisfied/pass}

		\addplot+[draw=white, fill=red!52]
		table[x=x,y=LostAtStop] {\errdata};
		\addlegendentry{lost at stop}

		\addplot+[draw=white, fill=orange!48]
		table[x=x,y=LiveAtStop] {\errdata};
		\addlegendentry{live at stop}

		\addplot+[draw=white, fill=purple!50]
		table[x=x,y=SatisfiedFailed] {\errdata};
		\addlegendentry{satisfied/fail}

	\end{axis}
\end{tikzpicture}
	\caption{Observable final outcomes across 619 tasks and four decoding modes. \textbf{Lost at stop}: final verdict \Lost. \textbf{Live at stop}: final verdict \Live{} without a closed derivation. \textbf{Satisfied/pass} and \textbf{/fail} split closed outputs by the grader of \S\ref{sec:gen}. The labels are observable final states, with no inference about suffix existence, and the constrained arms carry no lost-at-stop mass and more live-at-stop mass.}
	\Description{Four stacked bars show final outcomes for unconstrained, constrained direct, unconstrained thinking, and constrained mixed decoding. The two constrained modes have no Lost-at-stop mass and more Live-at-stop mass. Satisfied outputs are split into passing and failing outcomes.}
	\label{fig:errors}
\end{figure}

\subsection{Exploratory Observations}\label{sec:entropy}\label{sec:agents}

Figure~\ref{fig:entropy} shows the two rejection-telemetry observations.

\begin{figure}[t]
	\centering
\begin{tikzpicture}
	\begin{groupplot}[
			group style={group size=1 by 2, vertical sep=18mm},
			width=0.85\columnwidth,
			every axis/.append style={
					axis line style={draw=grid}, tick style={draw=grid},
					tick label style={font=\scriptsize, text=ink},
					label style={font=\footnotesize, text=ink},
					title style={font=\footnotesize, text=ink, yshift=-1mm},
				},
		]
		\nextgroupplot[
			height=4.2cm,
			title={\textbf{(a)} model uncertainty: entropy where the mask fires},
			ybar, bar width=7pt, bar shift auto,
			ymin=0, ymax=6, ytick={0,1,2,3,4,5,6},
			ylabel={pre-mask entropy (bits)},
			symbolic x coords={0.8B,2B,4B,9B,27B,gpt-oss},
			xtick=data, x tick label style={rotate=0},
			axis lines*=left, ymajorgrids, grid style={grid!55},
			legend style={draw=none, fill=none, font=\scriptsize, at={(0.02,0.98)}, anchor=north west},
			legend cell align=left,
		]
		\addplot+[draw=none, fill=green] coordinates
			{(0.8B,0.866)(2B,0.944)(4B,0.531)(9B,0.372)(27B,0.218)(gpt-oss,2.327)};
		\addlegendentry{token accepted}
		\addplot+[draw=none, fill=red] coordinates
			{(0.8B,1.705)(2B,1.673)(4B,1.056)(9B,1.284)(27B,0.812)(gpt-oss,5.506)};
		\addlegendentry{token rejected (retried)}

		\nextgroupplot[
			height=3.9cm,
			title={\textbf{(b)} grammar branching: tokens the oracle admits / position},
			xbar stacked, bar width=8.5pt,
			xmin=0, xmax=100, xtick={0,25,50,75,100},
			xlabel={share of positions (\%)},
			ytick={0,1,2,3}, yticklabels={{ML syntax}, {ML semantic}, {C syntax}, {C semantic}},
			y=5.2mm, enlarge y limits={abs=4.5mm},
			axis lines*=left,
			every axis plot/.append style={draw=white, line width=0.7pt},
			legend style={draw=none, fill=none, font=\scriptsize, legend columns=4,
					at={(0.5,-0.30)}, anchor=north, /tikz/every even column/.append style={column sep=3pt}},
			legend cell align=left,
		]
		\definecolor{tf}{HTML}{184F95}\definecolor{tt}{HTML}{2A78D6}\definecolor{to}{HTML}{86B6EF}\definecolor{tw}{HTML}{CDE2FB}
		\addplot+[fill=tf] coordinates {(9.09,0)(13.40,1)(6.02,2)(12.05,3)};
		\addlegendentry{forced (1)}
		\addplot+[fill=tt] coordinates {(2.39,0)(10.53,1)(13.25,2)(25.30,3)};
		\addlegendentry{tight (2--4)}
		\addplot+[fill=to] coordinates {(11.96,0)(13.88,1)(74.70,2)(62.65,3)};
		\addlegendentry{open (5--16)}
		\addplot+[fill=tw, nodes near coords={}] coordinates {(76.56,0)(62.20,1)(6.02,2)(0.00,3)};
		\addlegendentry{wide (17+)}
	\end{groupplot}
\end{tikzpicture}
	\caption{\textbf{(a)} Mean pre-mask entropy at steps where the sampler's first choice was accepted versus rejected (retried), per model, direct-mask arm: rejection is associated with the model's uncertain steps. \textbf{(b)} Share of positions at which the oracle admits one candidate token (forced), a few (tight), or many, semantic oracle against syntactic-only (14 ML / 5 C programs). }
	\Description{Two panels. Panel a is a grouped bar chart of pre-mask entropy in bits for six models. For each model the rejected-step bar is roughly twice the accepted-step bar or more. Panel b is a stacked horizontal bar chart with four rows, ML and C each under syntactic and semantic oracles, showing the distribution of admissible-set sizes binned as forced, tight, open, and wide. We see that semantic rows shift mass toward forced and tight bins.}
	\label{fig:entropy}
\end{figure}
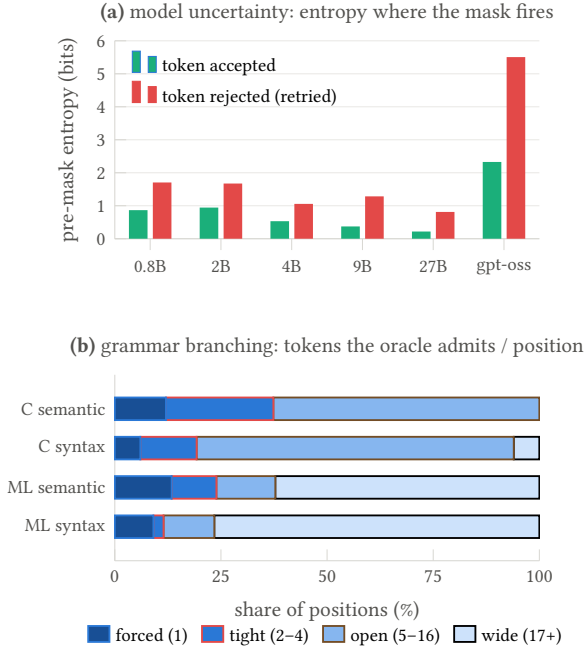

\paragraph{Rejection is associated with model uncertainty.}

Where the mask rejected the model's first choice, pre-mask entropy is markedly higher than at accepted steps, at every scale of the Qwen instruct ladder: 0.87 against 1.71 bits at 0.8B, 0.94 against 1.67 at 2B, 0.53 against 1.06 at 4B, 0.37 against 1.28 at 9B, 0.22 against 0.81 at 27B, and 2.33 against 5.51 on gpt-oss-20b (Figure~\ref{fig:entropy}a). This is an association in single greedy runs, not a predictive evaluation, but it matches CRANE's account of when constraint hurts~\cite{banerjee2025crane}: rejections concentrate at uncertain steps and are rare at confident ones. Retry pressure broadly falls with capability on that ladder (6.1\% of steps at 0.8B, 2.1\% at 27B, 13.9\% on 2B-Base) but not across families, gpt-oss-20b retries 32.9\%, while Pythia, the weakest model, retries only 6.8\%, so this is not a model ranking.

\paragraph{How strongly does the language constrain decoding?}
The rejection signal has a model-free counterpart: \emph{at each position of a valid program, how many candidate tokens does the oracle admit?}

We measure the mean $\log_2$ admissible-set size over a corpus-derived token alphabet: a Hartley-style branching measure under a uniform assumption, not Shannon entropy, conditioned on the evaluated corpus (14 ML programs, 209 positions, 42-token alphabet, and 5 C programs, 83 positions). Over the ML corpus the semantic oracle admits 16.8 candidate tokens per position on average against 26.3 for its syntactic twin, removing 9.5 candidates per position and raising the share of \emph{forced} positions, where exactly one candidate is legal, from 9\% to 13\% (Figure~\ref{fig:entropy}b). On C it admits 7.1 against 11.4. The typing layer thus closes roughly a third of the syntactic branching without ever rejecting the next gold token in the valid-corpus walk, and the measurement is model-free.

\begin{figure}[t]
	\centering
\begin{tikzpicture}
	\begin{groupplot}[
			group style={
					group size=2 by 1,
					horizontal sep=0.55cm,
				},
			width=0.36\columnwidth,
			height=4.3cm,
			scale only axis,
			ymin=0.5, ymax=7.5,
			ytick={7,6,5,4,3,2,1},
			axis lines*=left,
			axis line style={draw=grid},
			tick style={draw=grid},
			tick label style={font=\scriptsize, text=ink},
			label style={font=\footnotesize, text=ink},
			title style={font=\footnotesize\bfseries, text=ink, yshift=-1mm},
			xmajorgrids, grid style={grid!45},
			clip=false,
		]

		\nextgroupplot[
			xmode=log, log ticks with fixed point,
			xmin=20, xmax=3000,
			xtick={30,100,300,1000,3000},
			xlabel={mean generated tokens},
			title={Generation cost},
			yticklabels={Q3.5 0.8B, Q3.5 2B, Q3.5 4B, Q3.5 9B, Q3.6 27B, Phi-4-mini, gpt-oss-20b},
			yticklabel style={font=\scriptsize, text=ink, align=right},
			legend style={draw=none, fill=none, font=\scriptsize, legend columns=2, column sep=5pt, at={(1.12,1.18)}, anchor=south},
		]

		\addplot[draw=ink!38, line width=0.75pt, forget plot] coordinates {(2064.81,7) (195.13,7)};
		\addplot[draw=ink!38, line width=0.75pt, forget plot] coordinates {(1822.92,6) (85.69,6)};
		\addplot[draw=ink!38, line width=0.75pt, forget plot] coordinates {(2026.35,5) (33.90,5)};
		\addplot[draw=ink!38, line width=0.75pt, forget plot] coordinates {(1133.02,4) (35.69,4)};
		\addplot[draw=ink!38, line width=0.75pt, forget plot] coordinates {(1346.31,3) (28.04,3)};
		\addplot[draw=ink!38, line width=0.75pt, forget plot] coordinates {(538.52,2) (733.25,2)};
		\addplot[draw=ink!38, line width=0.75pt, forget plot] coordinates {(307.67,1) (269.06,1)};

		\addplot+[only marks, mark=o, mark size=2.2pt, color=purple, line width=0.9pt,
			mark options={draw=purple, fill=white, line width=0.9pt}]
		coordinates {(2064.81,7) (1822.92,6) (2026.35,5) (1133.02,4) (1346.31,3) (538.52,2) (307.67,1)};
		\addlegendentry{thinking}

		\addplot+[only marks, mark=*, mark size=2.2pt, color=orange]
		coordinates {(195.13,7) (85.69,6) (33.90,5) (35.69,4) (28.04,3) (733.25,2) (269.06,1)};
		\addlegendentry{mixed}

		\draw[grid!75, line width=0.5pt] (axis cs:20,2.5) -- (axis cs:3000,2.5);

		\nextgroupplot[
			xmin=10, xmax=100,
			xtick={25,50,75,100},
			xlabel={pass rate (\%)},
			title={Task performance},
			yticklabels={}, ytick style={draw=none},
		]

		\addplot[draw=ink!38, line width=0.75pt, forget plot] coordinates {(19.23,7) (25.00,7)};
		\addplot[draw=ink!38, line width=0.75pt, forget plot] coordinates {(71.15,6) (42.31,6)};
		\addplot[draw=ink!38, line width=0.75pt, forget plot] coordinates {(84.62,5) (80.77,5)};
		\addplot[draw=ink!38, line width=0.75pt, forget plot] coordinates {(88.46,4) (86.54,4)};
		\addplot[draw=ink!38, line width=0.75pt, forget plot] coordinates {(90.38,3) (86.54,3)};
		\addplot[draw=ink!38, line width=0.75pt, forget plot] coordinates {(19.23,2) (21.15,2)};
		\addplot[draw=ink!38, line width=0.75pt, forget plot] coordinates {(69.23,1) (82.69,1)};

		\addplot+[only marks, mark=o, mark size=2.2pt, color=purple, line width=0.9pt,
			mark options={draw=purple, fill=white, line width=0.9pt}]
		coordinates {(19.23,7) (71.15,6) (84.62,5) (88.46,4) (90.38,3) (19.23,2) (69.23,1)};

		\addplot+[only marks, mark=*, mark size=2.2pt, color=orange]
		coordinates {(25.00,7) (42.31,6) (80.77,5) (86.54,4) (86.54,3) (21.15,2) (82.69,1)};

		\draw[grid!75, line width=0.5pt] (axis cs:10,2.5) -- (axis cs:100,2.5);

	\end{groupplot}
\end{tikzpicture}
	\caption{Thinking vs.\ mixed mode: mean generated tokens on a log scale (left) and descriptive task-mix pass rate (right). Each row is one model and the line connects thinking (open purple) to mixed (filled orange). Mixed greatly reduces token count for the five Qwen rows (roughly $11\times$ to $60\times$), but not for Phi-4-mini ($0.73\times$, i.e.\ mixed emits more tokens) or gpt-oss-20b ($1.14\times$). Pass-rate changes are heterogeneous across rows.}
	\Description{Two aligned dot plots compare thinking and mixed decoding for seven models. Mixed uses far fewer tokens for five Qwen models, more for Phi-4-mini, and slightly fewer for gpt-oss-20b. Pass-rate changes vary by model.}
	\label{fig:efficiency}
\end{figure}
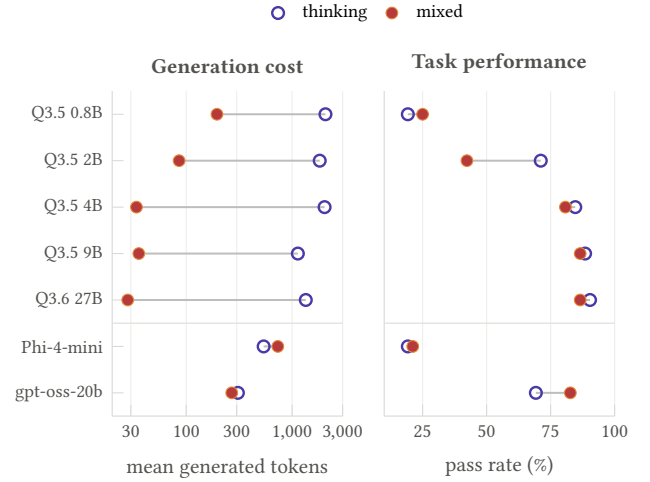

\paragraph{Agent episodes.}

Each turn decodes one registry step under the types established by prior turns, and episodes are graded by returned value. The comparison is unpaired (per-cell counts of 8 to 51), so we report direction only. Constrained Qwen3.5-4B-Base completes 57\% of episodes (8/14) against 22\% (4/18) unconstrained, while the fluent 4B is level (69\% against 71\%). Rates vary with surface form at least as much as with constraint: the fluent 4B reaches 88\% on S-expressions against 50\% on XML under the same typing rules. Several off-ladder models complete no episodes in either arm, and a paired design with matched coverage is future work.

\section{Related Work}\label{sec:related}

\paragraph{Syntactic constrained decoding.}

Outlines compiles regular constraints to automata over the token vocabulary~\cite{willard2023efficient}. PICARD prunes beam candidates by incremental parsing~\cite{scholak2021picard}, SynCode~\cite{ugare2025syncode}, Formatron~\cite{sun2025formatron}, and DOMINO~\cite{beurerkellner2024guiding} optimize grammar-aligned masking to near-zero overhead, and subword alignment has a clean automaton treatment~\cite{park2025flexible}. All decide membership in a fixed context-free or regular surface language. None carries the context-dependent state, such as scope, types, and declaration effects, on which models actually fail. Our syntactic layer uses the same Earley machinery. The contribution is the semantic layer above it and the contract analysis that survives the lift.

\paragraph{Semantic constraints.}

\textsc{ChopChop}~\cite{nagy2026chopchop} develops the general theory of semantic pruning for LLM decoding, casting admissibility as a realizability problem and stating soundness and completeness (their \S3.4) \emph{relative} to properties of the pruner the user programs: soundness needs an under-approximate or consistent checker, completeness an over-approximate and consistent one. For type safety they show no checker can be both, so their type pruner is over-approximate and consistent only to type depth three. We read their framework as the general setting our work instantiates.

We use a declarative class with a context-free surface and decidable first-order semantic premises. Safe pruning follows from stable contradictions. Dead-end freedom is a separate grammar property for which we give sufficient productivity, type-coverage, and left-to-right-flow conditions (Theorem~\ref{prop:complete}). The declarative route also gives a direct validation path: instances that check real languages can be tested prefix by prefix against production compilers (\S\ref{sec:cert}).

M\"undler et al.~\cite{muendler2025type} demonstrate type-constrained decoding for TypeScript with an incremental checker engineered for that language and evaluated at scale, a stronger per-language result than ours on its own territory. Our ML instance is a smaller monomorphic fragment, but comes from a declarative definition, and the same definition language yields the C-like fragment and the tool DSL with no new pruner code. IterGen~\cite{ugare2025itergen} backtracks over partial generations to enforce per-symbol constraints and Synchromesh~\cite{poesia2022synchromesh} enforces per-domain checks. Neither analyzes the two prefix contracts separately.

 Tractable control via distillation into a tractable model~\cite{zhang2023tractable} enforces constraints by construction rather than by pruning, a complementary route.

\paragraph{Grammar formalisms.}

The formalism assembles classical parts: attribute grammars with synthesized and inherited attributes~\cite{knuth1968semantics}, definite-clause grammars, which showed unification against a logic program's terms can carry a grammar's semantics~\cite{pereira1980dcg}, and unification-based grammar generally~\cite{shieber1986unification}.

  Semantic grammar specifications restrict these ideas to first-order terms, syntactic unification, and right-bound effects so that prefix verdicts are decidable and the pruning contract is analyzable. Condition (F) of \S\ref{sec:soundness} is the corresponding left-to-right restriction on a prefix parser.

\paragraph{Differential testing.}

Validating a checker against an independent implementation is differential testing in McKeeman's sense~\cite{mckeeman1998differential}, and its randomized form found hundreds of bugs in production C compilers~\cite{yang2011csmith}. Our variant tests \emph{prefix} behavior: the compiler grades whole programs, the oracle grades every prefix, and the harness checks the two never disagree in the direction that matters for decoding.

\paragraph{Reasoning and constraints.}

CRANE~\cite{banerjee2025crane} establishes that constrained and unconstrained decoding are usefully interleaved, reasoning outside the constrained block, the answer inside it, and explains why constraining from the first token degrades reasoning. Our mixed mode is that architecture with a semantic rather than syntactic block, and the telemetry of \S\ref{sec:entropy} agrees with its account: rejection lands on exactly the uncertain steps. SynCode~\cite{ugare2025syncode} reports the syntactic version.

\paragraph{Editors, holes, incremental typing.}

Typing incomplete programs is an editor-services problem before it is a decoding problem. Hazel gives incomplete programs a semantics in which the unfilled position is a first-class typed hole~\cite{omar2017hazelnut}. Our $\mathsf{Pending}$ binding is that object read through a parser, though $\Live$ is weaker than Hazel's typing, saying only that no stable contradiction has appeared (\S\ref{sec:soundness}). Incremental attribute evaluation with available/unavailable attributes goes back to Demers, Reps, and Teitelbaum~\cite{demers1981incremental}. Merlin~\cite{bour2018merlin} is the engineering dual, recovering from broken text to keep services alive where we prune extensions of well-typed prefixes: recovery maximizes tolerance, pruning maximizes rejection, and both need the same incremental frontier.

\section{Discussion and Limitations}\label{sec:discussion}

\paragraph{Formal limitations.} The safe-pruning proof is a paper proof of Aufbau's rejection discipline, not a mechanized proof of the code. Differential validation exercises the shipped implementation on finite cohorts, so implementation bugs outside those cohorts remain possible.
Dead-end freedom is proved only for grammars satisfying the sufficient conditions of \S\ref{sec:soundness}. We establish those conditions directly for the finite-lambda, core ML, and C-like fragments. We do not claim them for the tool DSLs, nor for the experimental STLC instance, whose unrestricted base-type universe defeats type coverage. The implementation does not attempt to infer this classification automatically. Token-sequence reachability is proved under exact spelling and the stated vocabulary-coverage hypothesis. Bounded proposal search additionally needs exhaustive masking, a deterministic fallback, or a separate sampler-fairness assumption which is assumed. Nothing is mechanized.


\paragraph{Empirical limitations.} The language metrics are heterogeneous: ML and C measure compiler-validated validity, STLC and tool tasks behavior. Generation uses one greedy attempt per task, the validation corpora are small and hand-selected, and the token-count and tool-episode comparisons are exploratory rather than causal. The C generation cells predate the return-type check now in \texttt{c.auf} (\S\ref{sec:instances}).

\section{Conclusion}\label{sec:conclusion}
A semantic prefix oracle owes the decoder two contracts, safe pruning and dead-end freedom. Aufbau provides safe pruning by rejecting only stable semantic contradictions. Dead-end freedom is separate. We give three sufficient grammar conditions, surface productivity, type coverage, and left-to-right constraint flow, and show how they separate the covered finite-lambda, ML, and C-like fragments from the uncovered experimental STLC and tool instances. A tokenizer-lifting lemma carries these guarantees from characters to token sequences under an explicit vocabulary-coverage hypothesis. Differential validation against production compilers shows zero false prunes across every prefix of 65 compiler-valid programs, and 25/30 invalid programs caught mid-stream where a syntax-only oracle catches none. These are finite compatibility results for the shipped engine, not an exhaustive correctness proof. In generation the matched ablation attributes positive observed lift to the typing rules on ML, while every matched C difference was zero, which format control alone can explain. The next steps are behavioral ML and C grading, generation experiments on the covered finite-lambda instance, a matched mixed-mode rerun, broader differential corpora, and mechanization.

\balance
\bibliographystyle{ACM-Reference-Format}
\bibliography{references}

\end{document}